\documentclass[journal,onecolumn,12pt]{IEEEtran}

 \usepackage{graphicx}
\usepackage{mathrsfs}
\usepackage{amsfonts}
\usepackage{graphicx}
\usepackage{amssymb}
\usepackage{amsmath}
\usepackage{cases}
\usepackage{caption}
\usepackage{hyperref}

\newtheorem{Thm}{Theorem}

\newtheorem{Lem}[Thm]{Lemma}

\newtheorem{Exp}{Example}
\newcommand{\wt}{{\mathrm{wt}}}
\newcommand{\tr}{{\mathrm{Tr}}}
\newcommand{\C}{{\mathcal{C}}}
\makeatletter
 \@addtoreset{equation}{section}
 
\makeatother

\begin{document}

\title{A Class of  Linear Codes With Three Weights\thanks{The  authors are supported by by a National Key Basic Research Project of China (2011CB302400), National Natural Science Foundation of China (61379139),  the ``Strategic Priority Research Program" of the Chinese Academy of Sciences, Grant No. XDA06010701 and Foundation of NSSFC(No.13CTJ006).}}

\author{Qiuyan Wang, \thanks{Q. Wang is with the State Key Laboratory of Information Security, the Institute of Information Engineering, The Chinese Academy of Sciences, Beijing, China. Email: wangqiuyan@iie.ac.cn}
 Fei Li$^{*}$\thanks{F. Li is the corresponding author and  with School of Statistics and Applied Mathematics,
   Anhui University of Finance and Economics,
 Bengbu City, Anhui Province, China. Email: cczxlf@163.com} and
 Dongdai Lin \thanks{ D. Lin is with the State Key Laboratory of Information Security, the Institute of Information Engineering, The Chinese Academy of Sciences, Beijing, China. Email: ddlin@iie.ac.cn}
}

\date{\today}
\maketitle

\begin{abstract}
Linear codes have been an interesting subject of study for many years. Recently, linear codes with few weights have been constructed and extensively studied.  In this paper, for an odd prime $p$, a class of three-weight linear codes over $\mathbb{F}_{p}$ are  constructed. The weight distributions of the linear codes are settled.  These codes have applications in  authentication codes, association schemes and data storage systems.
\end{abstract}


\section{Introduction}\label{sec-intro}

Throughout this paper, let $p$ be an odd prime, and let $q=p^{m}$ for a positive integer $m>2$. Let $\mathbb{F}_{q}$ and $\mathbb{F}_{p}$ denote the finite field with $q$ elements and $p$ elements, respectively.

An $(n,M)$ code $\mathcal{C}$ over $\mathbb{F}_{p}$ is a subset of $\mathbb{F}_{p}^{n}$ of size $M$. The vectors in $\mathcal{C}$ are called codewords of $\mathcal{C}$. The third important parameter of a code $\mathcal{C}$, besides the length $n$ and size $M$, is the minimum Hamming distance between codewords. The Hamming distance between two codewords $\mathbf{x},\mathbf{y}\in \mathcal{C}$ is defined to be the number of places where $\mathbf{x}$ and $\mathbf{y}$ differ. If the code $\mathcal{C}$ is a $k$-dimensional subspace of
 $\mathbb{F}_{p}^{n}$ with minimum (Hamming) distance $d$, it will be called an $[n,k,d]$ code over $\mathbb{F}_p$.

 Let $A_{i}$ be the number of codewords of weight $i$ in $\mathcal{C}$ of length $n$. The weight enumerator of $\mathcal{C}$ is defined by
$$
1+A_1x+A_2x^{2}+\cdots+A_{n}x^{n}.
$$
For $0\leq i\leq n$, the list $A_{i}$  is called the weight distribution or weight spectrum of $\mathcal{C}$. A code $\mathcal{C}$ is said to be a $t$-weight code if the number of nonzero $A_i$ with $1\leq i\leq n$ is equal to $t$. A great deal of research is devoted to the computation of the weight distribution of specific codes \cite{BM72,BM73,BM10,M66,M04,S95,SB12,V12}, since it does give important information of both practical and theoretical significance.

Let $D=\{d_1,d_2,\ldots, d_n\} \subseteq \mathbb{F}_{q}$. Let $\tr$ denote the trace function from $\mathbb{F}_q$ onto $\mathbb{F}_p$. A linear code of length $n$ over $\mathbb{F}_{p}$ is defined by
$$
\mathcal{C}_{D}=\{(\tr(xd_1), \tr(xd_2),\ldots, \tr(xd_{n})):x\in \mathbb{F}_{q}\},
$$
and $D$ is called the defining set of this code $\mathcal{C}_{D}$. This construction is proposed by Ding et al \cite{DLLZ} and  is generic in the sense that many known linear codes could be produced by selecting the defining set. If the defining set $D$ is well chosen, some optimal linear codes with few weights can be obtained \cite{DH,DYin,LY,WDX,ZD13,ZD14,ZLFH}. For more details, the readers are referred to \cite{D15,DD,DY,LYL1,LYL2}.

In this correspondence, for $a\in \mathbb{F}_{p}^{*}$, we set
$$
D_{a}=\{x\in \mathbb{F}_q^{*}: \tr(x)=a\}=\{d_1,d_2,\ldots, d_{n_a}\},
$$
and
\begin{equation}\label{defcode}
\mathcal{C}_{D_a}=\{\mathbf{c}_{x}=(\tr(xd_1^2), \tr(xd_2^2), \ldots, \tr(xd_{n_a}^2)): d_{i}\in D_a, 1\leq i\leq n_a, x\in \mathbb{F}_{q}\}.
\end{equation}

The objective of this paper is to determine the weight distribution of the proposed linear codes. Results show that they are three-weight linear codes. The  linear codes in this paper may yield associate schemes with framework introduced in \cite{CG} and can be employed to construct secret sharing schemes \cite{YD06}.

\section{Preliminaries}

In this section, we present some basic notations and results of group characters and exponential sums.

We start with the trace function. For $\alpha\in \mathbb{F}_{p^m}$, the \textit{absolute trace} $\tr(\alpha)$ of $\alpha$ is defined by \cite{LiNi}
$$
\tr(\alpha)=\alpha+\alpha^p+\alpha^{p^2}+\cdots+\alpha^{p^{m-1}}.
$$
By definition, $\tr(\alpha)$ is always an element of $\mathbb{F}_{p}$.

An additive character $\chi$ of $\mathbb{F}_{q}$ is a homomorphism from $\mathbb{F}_{q}$ into the multiplicative group $U$ of complex numbers of absolute value $1$, that is, a mapping from $\mathbb{F}_q$ into $U$ with $\chi(x+y)=\chi(x)\chi(y)$ for all $x,y \in \mathbb{F}_{q}$ \cite{LiNi}. For any $b\in \mathbb{F}_{q}$, the function
$$
\chi_{b}(x)=\zeta_{p}^{\tr(bx)} \textrm{for\ all }\ x\in \mathbb{F}_{q},
$$
defines an additive character of $\mathbb{F}_{q}$,
where $\zeta_{p}=e^{\frac{2\pi\sqrt{-1}}{p}}$. For $b=0$, the character $\chi_0(x)=1$ for all $x\in \mathbb{F}_{q}$ and is called the \textit{trivial} character of $\mathbb{F}_{q}$. All other additive character of $\mathbb{F}_{q}$ are called \textit{nontrivial}.    For $b=1$, the character $\chi_{1}$ will be called the \textit{canonical additive character} of $\mathbb{F}_{q}$. All additive characters of $\mathbb{F}_{q}$ can be expressed in terms of $\chi_{1}$: $\chi_{b}(x)=\chi_{1}(bx)$ for all $x\in \mathbb{F}_{q}$ \cite{LiNi}.

The orthogonal property of additive characters of $\mathbb{F}_{q}$ which can be found in Theorem $5.4$ in \cite{LiNi} is given by
\begin{equation}\label{eq-orthogonality}
\sum_{x\in \mathbb{F}_{q}}\chi(x)=\left\{\begin{array}{ll}
                                           q, & \textrm{if\ } \chi \textrm{\ is\ trivial},  \\
                                           0, & \textrm{if\ } \chi \textrm{\ is\ nontrivial}.
                                         \end{array}
                                         \right.
\end{equation}

Characters of the \textit{multiplicative group} $\mathbb{F}_{q}^{*}$ of $\mathbb{F}_{q}$ are called \textit{multiplicative characters} of $\mathbb{F}_{q}$. It is known that all characters of $\mathbb{F}_{q}^{*}$ are given by
$$
\psi_{j}(g^{k})=\zeta_{p}^{2\pi\sqrt{-1}jk/(q-1)} \textrm{for }\ k=0,1,\ldots, q-2,
$$
where $0\leq j\leq q-2$ and $g$ is a generator of $\mathbb{F}_{q}^{*}$ \cite{LiNi}. For $j=(q-1)/2$, the multiplicative character $\psi_{(q-1)/2}$ is called the \textit{quadratic character} of $\mathbb{F}_{q}$, and is denoted by $\eta$ in this paper. It is convenient to extend the definition of $\eta$ by setting $\eta(0)=0$. With this definition, we have then
\begin{equation}\label{eq-2.2}
\sum_{x\in \mathbb{F}_{q}}\eta(x)=\sum_{x\in \mathbb{F}_{q}^{*}}\eta(x)=0.
\end{equation}

The quadratic Gauss sum $G(\eta, \chi_1)$ over $\mathbb{F}_{q}$ is defined by
$$G(\eta, \chi_1)=\sum_{x\in \mathbb{F}_{q}}\eta(x)\chi_1(x),$$
and the quadratic Gauss sum $G(\overline{\eta}, \overline{\chi}_1)$ over $\mathbb{F}_{p}$ is defined by
$$
G(\overline{\eta},\overline{\chi}_{1})=\sum_{x\in \mathbb{F}_{p}}\overline{\eta}(x)\overline{\chi}_1(x),
$$
where $\overline{\eta}$ and $\overline{\chi}_1$ denote the quadratic and canonical character of $\mathbb{F}_{p}$, respectively.

The explicit values of quadratic Gauss sums over a finite field are determined and  given in the following lemma.
\begin{Lem}[Theorem $5.15$, \cite{LiNi}]\label{lem1}
Let the symbols be the same as before. Then
$$
G(\eta, \chi_1)=(-1)^{(m-1)}\sqrt{-1}^{\frac{(p-1)^2m}{4}}\sqrt{q},
$$
and
$$
G(\overline{\eta}, \overline{\chi}_1)=\sqrt{-1}^{\frac{(p-1)^{2}}{4}}\sqrt{p}.
$$

\end{Lem}
\begin{Lem} [Theorem $5.33$, \cite{LiNi}]\label{lem2}
Let $\chi$ be a nontrivial additive character of $\mathbb{F}_{q}$, and let $f(x)=a_2x^{2}+a_1x+a_0\in \mathbb{F}_{q}[x]$ with $a_2\neq0$. Then
$$
\sum_{x\in \mathbb{F}_{q}}\chi(f(x))=\chi(a_0-a_1^{2}(4a_2)^{-1})\eta(a_2)G(\eta,\chi).
$$
\end{Lem}

\begin{Lem}[Lemma 7, \cite{DD15}]\label{lem3}
Let the symbols be the same as before. Then

\begin{enumerate}
\item if $m\geq2$ is even, then $\eta(y)=1$ for each $y\in \mathbb{F}_{p}^{*}$; \\
\item if $m$ is odd, then $\eta(y)=\overline{\eta}(y)$ for each $y\in \mathbb{F}_{p}^{*}$.
\end{enumerate}
\end{Lem}
\begin{Lem}\label{lem4}
Let the symbols be the same as before. Then
$$
\sum_{y\in \mathbb{F}_{p}^{*}}\sum_{x\in \mathbb{F}_{q}}\chi_1(byx^{2})=\left\{\begin{array}{ll}
                                                                          0, & \textrm{if\ } m\equiv1\pmod{2}, \\
                                                                          (p-1)\eta(b)G(\eta, \chi_1), & \textrm{if\ } m\equiv0\pmod{2}.
                                                                        \end{array}
                                                                        \right.
$$
\end{Lem}
\begin{proof}
It follows from Lemma \ref{lem2} that
$$
\sum_{y\in \mathbb{F}_{p}^{*}}\sum_{x\in \mathbb{F}_{q}}\chi_1(byx^{2})=\eta(b)G(\eta,\chi_1)\sum_{y\in \mathbb{F}_{p}^{*}}\eta(y).
$$
Using Lemma \ref{lem3}, we get
$$
\sum_{y\in \mathbb{F}_{p}^{*}}\eta(y)=\left\{\begin{array}{ll}
                                               0, & \textrm{if\ } m\equiv1\pmod{2}, \\
                                               p-1,  & \textrm{if\ } m\equiv0\pmod{2}.
                                             \end{array}
                                             \right.
$$
Together with Lemma \ref{lem1}, we get this lemma.
\end{proof}
\section{The linear codes with three weights }
In this section, we will present a class of linear codes with three weights over $\mathbb{F}_{p}$. The weight distributions of the class linear codes are also settled.

For $a\in \mathbb{F}_{p}$,  set
 $$
 D_a=\{x\in \mathbb{F}_{q}^{*}: \tr(x)=a\}.
 $$
 If $a=0$, the weight distribution of the code $\C_{D_0}$ of \eqref{defcode} has been determined in \cite{YY15}. Hence, we only consider the case $a\in \mathbb{F}_{p}^{*}$ in this paper.

It is well known that  \cite{LiNi}
  \begin{equation}\label{eq-na}
  N_a=|\{x\in \mathbb{F}_{q}:\tr(x)=a, a\in \mathbb{F}_{p}\}|=p^{m-1}.
  \end{equation}

    Then  the length $n_a$ of the code  $\mathcal{C}_{D_a}$ $(a\in \mathbb{F}_{p}^{*})$ of \eqref{defcode} satisfies
 $$
 n_a= N_a=p^{m-1}.
 $$

 Define
 $$n_{(b,a)}=|\{x\in \mathbb{F}_{q}: \tr(x)=a \textrm{\ and\ } \tr(bx^2)=0\}|.$$
 For any $b\in \mathbb{F}_{q}^{*}$, the Hamming weight $\wt(\mathbf{c}_b)$ of the  codeword
 $
 \mathbf{c}_b
 $
 of the code $\mathcal{C}_{D_a}$ is given by
\begin{equation} \label{eq-wt}
\wt(\mathbf{c}_b)=N_a-n_{(b,a)}=p^{m-1}-n_{(b,a)}.
\end{equation}

 By   the orthogonal property of additive characters, for $b\in \mathbb{F}_{q}^{*}$ we have
\begin{align}\label{eq-weight}
  n_{(b,a)} &= p^{-2}\sum_{x\in \mathbb{F}_{q}}\left(\sum_{y \in F_{p}}\zeta_{p}^{y\tr(bx^{2})}\right)
\left(\sum_{z \in F_{p}}\zeta_{p}^{z(\tr(x)-a)}\right)  \nonumber \\
 &=p^{-2}\sum_{x\in \mathbb{F}_{q}}\left(1+\sum_{y \in F_{p}^{\ast}}\zeta_{p}^{y\tr(bx^{2})}\right)
\left(1+\sum_{z \in F_{p}^{\ast}}\zeta_{p}^{z(\tr(x)-a)}\right)\nonumber \\
&= p^{m-2} + p^{-2}\sum_{y \in F_{p}^{\ast}}\sum_{x\in \mathbb{F}_{q}}\zeta_{p}^{y\tr(bx^{2})}
+ p^{-2}\sum_{z \in F_{p}^{\ast}}\zeta_{p}^{-za}\sum_{x\in \mathbb{F}_{q}}\zeta_{p}^{z\tr(x)} \nonumber \\
&\quad +
p^{-2}\sum_{y \in F_{p}^{\ast}}\sum_{z \in F_{p}^{\ast}}\sum_{x\in \mathbb{F}_{q}}\zeta_{p}^{y\tr(bx^{2})+z\tr(x)-za} \nonumber\\
 &= p^{m-2} + p^{-2}\sum_{y \in F_{p}^{\ast}}\sum_{x\in \mathbb{F}_{q}}\zeta_{p}^{y\tr(bx^{2})} +
p^{-2}\sum_{y \in F_{p}^{\ast}}\sum_{z \in F_{p}^{\ast}}\zeta_{p}^{-za}\sum_{x\in \mathbb{F}_{q}}\zeta_{p}^{\tr(byx^{2}+zx)}
 \end{align}
In the sequel,  we will  calculate   $n_{(b,a)}$ for $a\in\mathbb{F}_p^{*}$ .

\begin{Lem}\label{lem5}
Let $b\in \mathbb{F}_{q}^{*}$. Then
\begin{align*}
&\sum_{y\in \mathbb{F}_{p}^{*}}\sum_{z\in \mathbb{F}_{p}^{*}}\sum_{x\in \mathbb{F}_{q}}\chi_1(byx^{2}+zx) & \\ &=\left\{\begin{array}{ll}
                                                                                                                   0,& \textrm{if\ } m\equiv1\pmod{2} \textrm{\ and\ } \tr(b^{-1})=0, \\
                                                                                                                   \eta(b)\eta(-\tr(b^{-1}))(-1)^{\frac{(p-1)(m+1)}{4}}(p-1)p^{\frac{m+1}{2}},&\textrm{if\ } m\equiv1\pmod{2} \textrm{\ and\ } \tr(b^{-1})\neq0,\\
                                                                                                                   \eta(b)(p-1)^{2}G(\eta,\chi_1),& \textrm{if\ } m\equiv0\pmod{2} \textrm{\ and\ } \tr(b^{-1})=0,
                                                                                                                   \\
                                                                                                                  -\eta(b)(p-1)G(\eta,\chi_1), & \textrm{if\ } m\equiv0\pmod{2}\textrm{\ and\ } \tr(b^{-1})\neq0,
                                                                                                                \end{array}
                                                                                                                \right. &
\end{align*}
where $G(\eta,\chi_1)$ is given in Lemma \ref{lem1}.
\end{Lem}
\begin{proof}
It follows from Lemmas \ref{lem2} and \ref{lem3} that
\begin{align*}
&\sum_{y\in \mathbb{F}_{p}^{*}}\sum_{z\in \mathbb{F}_{p}^{*}}\sum_{x\in \mathbb{F}_{q}}\chi_1(byx^{2}+zx) & \\
&=\sum_{y \in F_{p}^{\ast}}\sum_{z \in F_{p}^{\ast}}\chi_{1}\left(-\frac{z^{2}}{4by}\right)\eta(by)G(\eta, \chi_{1}) & \\
&= \eta(b)G(\eta, \chi_{1})\sum_{y_1 \in F_{p}^{\ast}}\sum_{z \in F_{p}^{\ast}}\chi_{1}\left(-\frac{y_1z^{2}}{b}\right)\eta\left(\frac{1}{4y_1}\right)& \\
& =\eta(b)G(\eta, \chi_{1})\sum_{y \in F_{p}^{\ast}}\sum_{z \in F_{p}^{\ast}}\chi_{1}\left(-\frac{yz^{2}}{b}\right)\eta(y) &\\
& =\eta(b)G(\eta, \chi_{1})\sum_{z \in F_{p}^{\ast}}\sum_{y \in F_{p}^{\ast}}\zeta_{p}^{-yz^{2}\tr(b^{-1})}\eta(y) & \\
& =  \left \{
   \begin{array}{ll}
  \eta(b)G(\eta, \chi_{1})\sum_{z \in F_{p}^{\ast}}\sum_{y \in F_{p}^{\ast}}\eta(y) , & \textrm{if\ } \tr(b^{-1})=0, \\
  \eta(b)\eta(-\tr(b^{-1}))G(\eta, \chi_{1})\sum_{z \in F_{p}^{\ast}}\sum_{y \in F_{p}^{\ast}}
  \zeta_{p}^{yz^{2}\tr(b^{-1})}\eta\left(yz^{2}\tr(b^{-1})\right) ,  & \textrm{if\ } \tr(b^{-1})\neq 0.
\end{array}
  \right. \\
&=\left \{
   \begin{array}{ll}
  \eta(b)G(\eta, \chi_{1})\sum_{z \in F_{p}^{\ast}}\sum_{y \in F_{p}^{\ast}}\eta(y) , &\textrm{if\ }  \tr(b^{-1})=0, \\
  \eta(b)\eta\left(-
  \tr(b^{-1})\right)G(\eta, \chi_{1})\sum_{z \in F_{p}^{\ast}}\sum_{y \in F_{p}^{\ast}}
  \zeta_{p}^{y}\eta(y) ,&  \textrm{if } \tr(b^{-1})\neq 0.
\end{array}
  \right.\\
  &=\left\{\begin{array}{ll}
             0, & \textrm{if } m\equiv1\pmod{2} \textrm{\ and\ } \tr(b^{-1})=0, \\
             \eta(b)\eta(-\tr(b^{-1}))(p-1)G(\eta,\chi_1)G(\overline{\eta}, \overline{\chi}_1) & \textrm{if\ } m\equiv1\pmod{2} \textrm{\ and\ } \tr(b^{-1})\neq0, \\
             \eta(b)(p-1)^{2}G(\eta, \chi_1), & \textrm{if\ } m\equiv0\pmod{2}\textrm{\ and\ } \tr(b^{-1})=0, \\
             -\eta(b)(p-1)G(\eta, \chi_1), & \textrm{if\ } m\equiv0\pmod{2} \textrm{\ and\ } \tr(b^{-1})\neq0.
           \end{array}
           \right.
\end{align*}

Then the desired conclusion follows Lemma \ref{lem1}.
\end{proof}
\begin{Lem}\label{lem6}
Let
$$M=\{b\in \mathbb{F}_{q}^{*}: \eta(b)=-1\textrm{\ and\ } \tr(b)=0\}$$
and
$$N=\{b\in \mathbb{F}_{q}^{*}: \eta(b)=1\textrm{\ and\ } \tr(b)=0\}.$$
Then
\begin{enumerate}
\item  $$|M|=\left\{\begin{array}{ll}
                                                                           \frac{1}{2p}(q-p), & \textrm{if\ } m\equiv1\pmod{2}, \\
                                                                           \frac{1}{2p}(q-p)+\frac{1}{2}(-1)^{\frac{(p-1)m}{4}}(p-1)p^{\frac{m-2}{2}},  & \textrm{if\ } m\equiv0\pmod{2}.
                                                                         \end{array}
                                                                         \right.$$

\item  $$
|N|=\left\{\begin{array}{ll}
         \frac{1}{2p}(q-p),   & \textrm{if\ } m\equiv1\pmod{2},  \\
        \frac{1}{2p}(q-p)-\frac{1}{2}(-1)^{\frac{(p-1)m}{4}}(p-1)p^{\frac{m-2}{2}},    & \textrm{if\ } m\equiv0\pmod{2}.
         \end{array}
         \right.
$$
\end{enumerate}
\end{Lem}

\begin{proof}
 We only prove the first part of this lemma,  since $|M|+|N|=p^{m-1}-1$.

It follows from \eqref{eq-orthogonality}, \eqref{eq-2.2}, Lemmas \ref{lem1} and \ref{lem3} that
\begin{align*}
|M|&=\frac{1}{2p}\sum_{x \in F_{q}^{\ast}}\left(1-\eta(x)\right)\left(\sum_{y \in F_{p}}\zeta_{p}^{y\tr(x)}\right)& \\
& =\frac{1}{2p}\sum_{x \in F_{q}^{\ast}}\left(1-\eta(x)+\sum_{y \in F_{p}^{\ast}}\zeta_{p}^{y\tr(x)}-
\sum_{y \in F_{p}^{\ast}}\eta(x)\zeta_{p}^{y\tr(x)}\right) & \\
& =\frac{1}{2p}(q-1)-\frac{1}{2p}\sum_{x \in F_{q}^{\ast}}\eta(x)+\frac{1}{2p}\sum_{x \in F_{q}^{\ast}}\sum_{y \in F_{p}^{\ast}}\zeta_{p}^{y\tr(x)}-
\frac{1}{2p}\sum_{x \in F_{q}^{\ast}}\sum_{y \in F_{p}^{\ast}}\eta(x)\zeta_{p}^{y\tr(x)}& \\
& =\frac{1}{2p}(q-1)+\frac{1}{2p}\sum_{y \in F_{p}^{\ast}}\left(\sum_{x \in F_{q}}\zeta_{p}^{y\tr(x)}-1\right)-
\frac{1}{2p}\sum_{y \in F_{p}^{\ast}}\eta(y)\sum_{x \in F_{q}^{\ast}}\eta(yx)\zeta_{p}^{y\tr(x)} & \\
& =\frac{1}{2p}(q-1)-\frac{1}{2p}(p-1)-
\frac{1}{2p}G(\eta, \chi_{1})\sum_{y \in F_{p}^{\ast}}\eta(y) & \\
&=  \left \{
   \begin{array}{ll}
   \frac{1}{2p}(q-p), &\textrm{if\ }  m \equiv 1 \mod2, \\
 \frac{1}{2p}(q-p)+(-1)^{\frac{m(p-1)}{4}}\frac{1}{2}(p-1)p^{\frac{m-2}{2}},& \textrm{if\ } m \equiv 0 \mod2.
\end{array}
  \right.
\end{align*}
This completes the proof of the first conclusion of this lemma.
\end{proof}
\begin{Lem}\label{lem7}
For each $a\in \mathbb{F}_{p}^{*}$, let
$$
M_a=\{x\in \mathbb{F}_{q}: \eta(x)=-1 \textrm{\ and \ }\tr(x)=a\}.
$$
Then
$$
|M_a|=\left\{\begin{array}{ll}
             \frac{q}{2p}-\frac{1}{2}\eta(-a)(-1)^{\frac{(m+1)(p-1)}{4}}p^{\frac{m-1}{2}}, & \textrm{if\ } m\equiv1\pmod{2}, \\
             \frac{q}{2p}-\frac{1}{2}(-1)^{\frac{(p-1)m}{4}}p^{\frac{m-2}{2}},  & \textrm{if\ }  m\equiv0\pmod{2}.
           \end{array}
           \right.
$$
\end{Lem}
\begin{proof}
By \eqref{eq-orthogonality}, \eqref{eq-2.2} and Lemma \ref{lem1},  for any $a\in \mathbb{F}_{p}^{*}$ we have
\begin{align*}
 |M_{a}| &=\frac{1}{2p}\sum_{x \in F_{q}}\left(1-\eta(x)\right)\left(\sum_{y \in F_{p}}\zeta_{p}^{y(\tr(x)-a)}\right) & \\
  & =\frac{1}{2p}\sum_{x \in F_{q}}(1-\eta(x))+\frac{1}{2p}\sum_{x \in F_{q}}\sum_{y \in F_{p}^{\ast}}\zeta_{p}^{y(\tr(x)-a)}-\frac{1}{2p}
\sum_{x \in F_{q}}\sum_{y \in F_{p}^{\ast}}\eta(x)\zeta_{p}^{y(\tr(x)-a)} &\\
 & =\frac{q}{2p}-\frac{1}{2p}\sum_{x \in F_{q}}\eta(x)+\frac{1}{2p}\sum_{y \in F_{p}^{\ast}}\sum_{x \in F_{q}}\zeta_{p}^{y(\tr(x)-a)}-
\frac{1}{2p}\sum_{y \in F_{p}^{\ast}}\sum_{x \in F_{q}}\eta(x)\zeta_{p}^{y(\tr(x)-a)} & \\
 &=\frac{q}{2p}-
\frac{1}{2p}\sum_{y \in F_{p}^{\ast}}\eta(y)\zeta_{p}^{-ya}\sum_{x \in F_{q}^{\ast}}\eta(yx)\zeta_{p}^{y\tr(x)}  &\\
 &=\frac{q}{2p}-
\frac{1}{2p}\eta(-a)G(\eta, \chi_{1})\sum_{y \in F_{p}^{\ast}}\eta(-ay)\zeta_{p}^{-ya} & \\
 &=\left \{
   \begin{array}{ll}
   \frac{q}{2p}-
\frac{1}{2p}\eta(-a)G(\eta, \chi_{1})G(\overline{\eta}, \overline{\chi}_{1}),&  \textrm{if\ }  m \equiv 1 \mod2, \\
 \frac{q}{2p}+\frac{1}{2p}G(\eta, \chi_{1}),  & \textrm{if\ }  m \equiv 0 \mod2.
\end{array}
  \right.\\
 & =\left\{\begin{array}{ll}
            \frac{q}{2p}-\frac{1}{2}(-1)^{\frac{(m+1)(p-1)}{4}}\eta(-a)p^{\frac{m-1}{2}}, & \textrm{if\ } m\equiv1\pmod{2}, \\
            \frac{q}{2p}-\frac{1}{2}(-1)^{\frac{(p-1)m}{4}}p^{\frac{m-2}{2}},  & \textrm{if\ }  m\equiv0\pmod{2}.
          \end{array}
          \right.
\end{align*}
This completes the proof of this lemma.
\end{proof}
\begin{Lem}\label{lem8}
Let the symbols be the same as before. For each $b\in \mathbb{F}_{q}^{*}$, we have
\begin{enumerate}

\item if $m$ is even, then
$$
n_{(b,0)}=\left\{\begin{array}{ll}
             p^{m-2}-(-1)^{\frac{(p-1)m}{4}}(p-1)\eta(b)p^{\frac{m-2}{2}},  & \textrm{if\ }\tr(b^{-1})=0,  \\
             p^{m-2},  & \textrm{if\ } \tr(b^{-1})\neq0;
           \end{array}
           \right.
$$
\item if $m$ is odd, then
$$
n_{(b,0)}=\left\{\begin{array}{ll}
             p^{m-2},  & \textrm{if\ } \tr(b^{-1})=0, \\
             p^{m-2}+(-1)^{\frac{(p-1)(m+1)}{4}}(p-1)\eta(b)\eta\left(-\tr(b^{-1})\right)p^{\frac{m-3}{2}},  & \textrm{if\ } \tr(b^{-1})\neq0.
           \end{array}
           \right.
$$
\end{enumerate}
\end{Lem}
\begin{proof}
From the equation \eqref{eq-weight}, we get
\begin{align*}
n_{(b,0)}=p^{m-2} + p^{-2}\sum_{y \in F_{p}^{\ast}}\sum_{x \in F_{q}}\zeta_{p}^{y\tr(bx^{2})}
+ p^{-2}\sum_{y \in F_{p}^{\ast}}\sum_{z \in F_{p}^{\ast}}\sum_{x \in F_{q}}\zeta_{p}^{y\tr(bx^{2})+z\tr(x)}.  \\
\end{align*}

The desired conclusions then follow from Lemmas  \ref{lem4} and \ref{lem5}.
\end{proof}
Next, we introduce the  class of linear codes of \eqref{defcode} with three weights.
For $a\in \mathbb{F}_{p}^{*}$, put
\begin{equation}\label{def-d2}
D_a=\{x\in \mathbb{F}_{q}^{*}: \tr(x)=a\}.
\end{equation}

To get the weight distribution of $\mathcal{C}_{D_a}$ constructed from the  set $D_a$ of \eqref{def-d2}, we need a number of auxiliary results.

\begin{Lem}\label{lem12}
For any $a\in \mathbb{F}_{p}^{*}$, let
$$ A=\sum_{y \in F_{p}^{\ast}}\sum_{z \in F_{p}^{\ast}}
\zeta_{p}^{-za}\sum_{x\in \mathbb{F}_{q}}\zeta_{p}^{\tr(byx^{2}+zx)}. $$
Then for $ b \in F_{q}^{\ast}$ we have
$$
A= \left \{
   \begin{array}{ll}
   0,& \text{if\ } m \equiv 1 \pmod2, \tr(b^{-1})=0, \\
    -(-1)^{\frac{(m+1)(p-1)}{4}}\eta(b)\eta(-\tr(b^{-1}))p^{\frac{m+1}{2}},
    &  \text{if\ } m \equiv 1 \pmod2, \tr(b^{-1}) \neq 0, \\
    (-1)^{(\frac{p-1}{2})^{2}\frac{m}{2}}\eta(b)(p-1)p^{\frac{m}{2}}, & \text{if\ } m \equiv 0 \pmod2, \tr(b^{-1})=0, \\
     -(-1)^{(\frac{p-1}{2})^2\frac{m}{2}}\eta(b)p^{\frac{m}{2}}, & \text{if\ }m \equiv 0 \pmod2, \tr(b^{-1}) \neq 0,
\end{array}
  \right.
  $$
  where $G(\eta, \chi_1)$ and $G(\overline{\eta}, \overline{\chi}_{1})$ denote the square Gauss sum over $\mathbb{F}_{q}$ and $\mathbb{F}_{p}$, respectively.
\end{Lem}
\begin{proof}
By Lemma \ref{lem2}, we have that
\begin{align*}
A&=\sum_{y \in F_{p}^{\ast}}\sum_{z \in F_{p}^{\ast}}\zeta_{p}^{-za}\chi_{1}\left(-\frac{z^{2}}{4by}\right)\eta(by)G(\eta, \chi_{1})& \\
&=\eta(b)G(\eta, \chi_{1})\sum_{y \in F_{p}^{\ast}}\sum_{z \in F_{p}^{\ast}}\zeta_{p}^{-za}\chi_{1}\left(-\frac{yz^{2}}{b}\right)\eta\left(\frac{1}{4y}\right)& \\
&=\eta(b)G(\eta, \chi_{1})\sum_{y \in F_{p}^{\ast}}\sum_{z \in F_{p}^{\ast}}\zeta_{p}^{-za}\chi_{1}\left(-\frac{yz^{2}}{b}\right)\eta(y) &\\
&=\eta(b)G(\eta, \chi_{1})\sum_{z \in F_{p}^{\ast}}\zeta_p^{-za}\sum_{y \in F_{p}^{\ast}}\zeta_{p}^{-yz^{2}\tr(b^{-1})}\eta(y) & \\
&=  \left \{
   \begin{array}{ll}
  \eta(b)G(\eta, \chi_{1})\sum_{z \in F_{p}^{\ast}}\zeta_{p}^{-za}\sum_{y \in F_{p}^{\ast}}\eta(y) ,&\textrm{if\ }\tr(b^{-1})=0, \\
  \eta(b)\eta(-\tr(b^{-1}))G(\eta, \chi_{1})\sum_{z \in F_{p}^{\ast}}\zeta_{p}^{-za}\sum_{y \in F_{p}^{\ast}}
  \zeta_{p}^{yz^{2}\tr(b^{-1})}\eta\left(yz^{2}\tr(b^{-1})\right) ,&  \textrm{if\ } \tr\left(b^{-1}\right)\neq 0.
\end{array}
  \right.  \\
&= \left \{
   \begin{array}{ll}
  \eta(b)G(\eta, \chi_{1})\sum_{z \in F_{p}^{\ast}}\zeta_{p}^{-za}\sum_{y \in F_{p}^{\ast}}\eta(y) ,& \textrm{if\ } \tr(b^{-1})=0, \\
  \eta(b)\eta(-\tr(b^{-1}))G(\eta, \chi_{1})\sum_{z \in F_{p}^{\ast}}\zeta_{p}^{-za}\sum_{y \in F_{p}^{\ast}}
  \zeta_{p}^{y}\eta(y) ,& \textrm{if\ }\tr(b^{-1})\neq 0.
\end{array}
  \right.
  \end{align*}
 For $a\in \mathbb{F}_{p}^{*}$,  it is easy to check that
  $$ \sum_{z \in F_{p}^{\ast}}\zeta_{p}^{-za}=-1. $$
Then the results follow from Lemmas \ref{lem1} and \ref{lem3}.
\end{proof}
The following lemma follows directly from \eqref{eq-weight}, Lemmas \ref{lem4} and \ref{lem12}.
\begin{Lem}\label{lem13}
Let the symbols be the same as before. Then for $a\in \mathbb{F}_{p}^{*}$ and $b\in \mathbb{F}_{q}^{*}$
$$
n_{(b,a)}=\left\{\begin{array}{ll}
                   p^{m-2}, & \textrm{if\ } \tr(b^{-1})=0, \\
                   p^{m-2}-(-1)^{\frac{(m+1)(p-1)}{4}}\eta(b)\eta(-\tr(b^{-1}))p^{\frac{m-3}{2}}, , & \textrm{if\ }  m\equiv1\pmod{2},  \tr(b^{-1})\neq0, \\
                   p^{m-2}-(-1)^{(\frac{p-1}{2})^{2}\frac{m}{2}}p^{\frac{m-2}{2}}\eta(b), & \textrm{if\ } m\equiv0\pmod{2}, \tr(b^{-1})\neq0.
                 \end{array}
                 \right.
$$
\end{Lem}

By Lemmas \ref{lem7} and \ref{lem13}, if $m>2$,  for $a\in \mathbb{F}_{p}^{*}$ we obtain
$$ n_{(b,a)} \in \left\{p^{m-2},
p^{m-2}-p^{\frac{m-2}{2}},p^{m-2}+p^{\frac{m-2}{2}}\right\} $$
or  $$ n_{(b,a)} \in \left\{p^{m-2},
p^{m-2}-p^{\frac{m-3}{2}},p^{m-2}+p^{\frac{m-3}{2}}\right\} $$
according as $m$ is even or odd, respectively.

Now, for $a\in \mathbb{F}_{p}^{*}$ and the set $D_a$ of \eqref{def-d2},  we are ready to determine the weight distribution of the code $\mathcal{C}_{D_a}$ of \eqref{defcode}.

\begin{table}[ht]
\centering
\caption{The weight distribution of the codes of Theorem \ref{Thm3}}\label{tal:weightdistribution3}
\begin{tabular}{|c|c|}
\hline
\textrm{Weight}  \qquad& \textrm{Multiplicity}    \\
\hline
0 \qquad&   1  \\
\hline
$(p-1)p^{m-2}$ \qquad&  $p^{m-1}-1$  \\
\hline
$(p-1)p^{m-2}+p^{\frac{m-3}{2}}$  \qquad& $\frac{1}{2}(p-1)(p^{m-1}+p^{\frac{m-1}{2}})$ \\
\hline
$(p-1)p^{m-2}-p^{\frac{m-3}{2}}$ \qquad&  $\frac{1}{2}(p-1)(p^{m-1}-p^{\frac{m-1}{2}})$  \\
\hline
\end{tabular}
\end{table}
\begin{Thm}\label{Thm3}
For $a\in \mathbb{F}_{p}^{*}$, if $m$ is odd, then the linear code $\mathcal{C}_{D_a}$ over $\mathbb{F}_{p}$ has parameters $[p^{m-1}, m]$  and weight distribution in \autoref{tal:weightdistribution3}.
\end{Thm}
\begin{proof}
If $m\equiv1\pmod{2}$, by Lemma \ref{lem7}, for $a\in \mathbb{F}_{p}^{*}$ we have that
\begin{align*}
 |M_{a}|&=\left|\{ x \in F_{q}: \eta(x)=-1 \textrm{\ and\ } \tr(x)=a \}\right|\\
 &=\frac{q}{2p}-
\frac{1}{2p}\eta(-a)G(\eta, \chi_{1})G(\overline{\eta}, \overline{\chi}_{1}).
\end{align*}
 Let $$ \overline{M_{a}}=\left\{ x \in F_{q}: \eta(x)=1 \textrm{\ and\ } \tr(x)=a \right\}, $$
$$ S=\{ x \in F_{q}: \eta(x)=-1 \textrm{\ and\ } \eta\left(\tr(x)\right)=1 \}, $$ and $$ \overline{S}=\{ x \in F_{q}: \eta(x)=1 \textrm{\ and\ } \eta\left(\tr(x)\right)=-1 \}. $$
Notice that half of the elements in $ F_{p}^{\ast} $ are squares. Hence we get
$$ |S|=\frac{p-1}{2}\left(\frac{q}{2p}-
\frac{1}{2p}\eta(-1)G(\eta,\chi_1)G(\overline{\eta},\overline{\chi}_1)\right) $$ and
\begin{align*}
 |\overline{S}|=\frac{p-1}{2}\left(p^{m-1}-\frac{q}{2p}+
\frac{1}{2p}\eta(-1)(-1)G(\eta, \chi_{1})G(\overline{\eta}, \overline{\chi}_{1})\right).
\end{align*}
Set $$ T=\{ x \in F_{q}: \eta(x)\eta( \tr(x))=-1 \},$$ and $$\overline{T}=\{ x \in F_{q}: \eta(x)\eta( \tr(x))=1 \}. $$
By definition,  we obtain
 $$ |T|=|S|+|\overline{S}|=\frac{p-1}{2}\left(p^{m-1}-\frac{1}{p}\eta(-1)G(\eta,\chi_1)G(\overline{\eta}, \overline{\chi}_1)\right). $$
and
$$ |\overline{T}|=p^{m}-p^{m-1}-|T|= \frac{p-1}{2}\left(p^{m-1}+\frac{1}{p}\eta(-1)G(\eta,\chi_1)G(\overline{\eta},\overline{\chi}_1)\right)$$
Note that $ \eta(-1)=(-1)^{\frac{p-1}{2}}$
and $ G(\eta, \chi_{1})G(\overline{\eta}, \overline{\chi}_{1})=(-1)^{\frac{(m+1)(p-1)}{4}}p^{\frac{m+1}{2}}. $
Then the weight distribution of \autoref{tal:weightdistribution3} follows from  \eqref{eq-wt} and  Lemma \ref{lem13}. It can be easily checked that $\wt(\mathbf{c}_b)>0$ for $b\in \mathbb{F}_{q}^{*}$. Hence, the dimension of this code $\mathcal{C}_{D_a}$ of Theorem \ref{Thm3} is equal to $m$.
\end{proof}

\begin{Exp}
Let $(p,m,a)=(5,3,1)$. Then the corresponding code $\mathcal{C}_{D_1}$ has parameters $[ 25,3,19]$ and weight enumerator
$1+40x^{19}+24x^{20}+60x^{21}$. Remark that this code is almost optimal, since an optimal  $[25,3]$ code has minimum distance $20$.
\end{Exp}
\begin{Exp}
Let $(p,m,a)=(3,3,1)$. Then the corresponding code $\mathcal{C}_{D_1}$ has parameters $[9,3,5]$ and weight enumerator $1+6x^{5}+8x^{6}+12x^7$.  Remark that this code is almost optimal, since  an optimal  $[9,3]$ code has minimum distance $6$.
\end{Exp}
\begin{table}[ht]
\centering
\caption{The weight distribution of the codes of Theorem \ref{Thm4}}\label{tal:weightdistribution4}
\begin{tabular}{|c|c|}
\hline
\textrm{Weight}  \qquad& \textrm{Multiplicity}    \\
\hline
0 \qquad&   1  \\
\hline
$(p-1)p^{m-2}$ \qquad&  $p^{m-1}-1$  \\
\hline
$(p-1)p^{m-2}+p^{\frac{m-2}{2}}$  \qquad& $\frac{1}{2}(p-1)(p^{m-1}+p^{\frac{m-2}{2}})$ \\
\hline
$(p-1)p^{m-2}-p^{\frac{m-2}{2}}$ \qquad&  $\frac{1}{2}(p-1)(p^{m-1}-p^{\frac{m-2}{2}})$   \\
\hline
\end{tabular}
\end{table}

\begin{Thm}\label{Thm4}
For $a\in \mathbb{F}_{p}^{*}$, if $m$ is even, then the linear code $\mathcal{C}_{D_a}$ over $\mathbb{F}_{p}$ has parameters $[p^{m-1}, m]$  and weight distribution in \autoref{tal:weightdistribution4}.
\end{Thm}
\begin{proof}
 If $m \equiv0 \pmod{2},$ by Lemma \ref{lem6}, we have
\begin{align*}
  |M|&=\left|\{b\in \mathbb{F}_{q}^{*}: \eta(b)=-1\textrm{\ and\ } \tr(b)=0\}\right|&\\
  &=\frac{1}{2p}(q-p)+\frac{1}{2p}(-1)^{\frac{m(p-1)}{4}}(p-1)p^{\frac{m}{2}}.&.
\end{align*}
Let
\begin{align*}
 \overline{M}&=\{b\in \mathbb{F}_{q}^{*}: \eta(b)=-1\textrm{\ and\ } \tr(b)\neq 0\},\\
N&=\{b\in \mathbb{F}_{q}^{*}: \eta(b)=1\textrm{\ and\ } \tr(b)= 0\},
\end{align*}
 and
$$\overline{N}=\{b\in \mathbb{F}_{q}^{*}: \eta(b)=1\textrm{\ and\ } \tr(b)\neq 0\}. $$
By definition, we know
$$|N|=p^{m-1}-1-|M| .$$
 Since $|\{b\in \mathbb{F}_{q}^{*}: \eta(b)=-1\}|=(q-1)/2$, we get
\begin{align*}
|\overline{M}|=\frac{1}{2}(q-1)-|M|=\frac{p-1}{2}\left(p^{m-1}-(-1)^{\frac{m(p-1)}{4}}p^{\frac{m-2}{2}}\right),
\end{align*}
and
\begin{align*}
 |\overline{N}|=\frac{1}{2}(q-1)-|N|=\frac{p-1}{2}\left(p^{m-1}+(-1)^{\frac{(p-1)m}{4}}p^{\frac{m-2}{2}}\right).
\end{align*}
The weight distribution of \autoref{tal:weightdistribution4} follows from  \eqref{eq-wt} and  Lemma \ref{lem13}. The dimension of the code $\mathcal{C}_{D_a}$ of Theorem \ref{Thm4} equals $m$, since $\wt(\mathbf{c}_b)>0$ for $b\in \mathbb{F}_{q}^{*}$.
\end{proof}
\begin{Exp}
Let $(p,m,a)=(5,4,1)$. Then the corresponding code $\mathcal{C}_{D_1}$ has parameters $[125,4,95]$ and weight enumerator $1+240x^{95}+124x^{100}+260x^{105}$.
\end{Exp}
\begin{Exp}
Let $(p,m,a)=(3,4,1)$. Then the corresponding code $\mathcal{C}_{D_1}$ has parameters $[27, 3, 15 ]$ and weight enumerator $1+24x^{15}+26x^{18}+30x^{21}$.
\end{Exp}
\section{Concluding Remarks}
In this paper, we present a class of linear codes with three weights. There is a survey on three-weight codes in \cite{DLLZ}. A number of three-weight codes were constructed in \cite{CK,CW,Ding09,D15,DY,DD15,DD,LYL1,LYL2,WDX,YY15,ZD14,ZLFH}. We did not find the parameters of the class of  linear codes $\mathcal{C}_{D_a}$ ($a\in \mathbb{F}_{p}^{*}$) in this paper in these references.

Let $w_{\min}$ and $w_{\max}$ denote the minimum and maximum nonzero weight of a linear code $\mathcal{C}$. As stated in \cite{YD06}, the linear code $\mathcal{C}$ can be employed to construct a secret sharing scheme with  interesting access structures if  $w_{\min}/w_{\max}>p-1/p$.

Let $m\equiv1\pmod{2}$ and $m>3$. Then for the linear code $\mathcal{C}_{D_a}$ ($a\neq0$) of Theorem \ref{Thm3}, we have
$$
\frac{w_{\min}}{w_{\max}}=\frac{(p-1)p^{m-2}-p^{\frac{m-3}{2}}}{(p-1)p^{m-2}+p^{\frac{m-3}{2}}}>\frac{p-1}{p}.
$$

Let $m\equiv0\pmod{2}$ and $m>2$. Then for the linear code $\mathcal{C}_{D_a}$ ($a\neq0$) of Theorem \ref{Thm4}, we have
$$
\frac{w_{\min}}{w_{\max}}=\frac{(p-1)p^{m-2}-p^{\frac{m-2}{2}}}{(p-1)p^{m-2}+p^{\frac{m-2}{2}}}>\frac{p-1}{p}.
$$

Hence, the linear codes in this paper satisfy $w_{\min}/w_{\max}>(p-1)/p$ if $m\geq6$, and can be used to get secret sharing schemes with interesting access structures.

\end{document}